\documentclass[12pt]{article}
\usepackage{graphicx}
\usepackage{amssymb}
\usepackage{amsmath}
\usepackage{amsfonts}
\usepackage{amsthm}
\usepackage[english]{babel}
\usepackage[latin1,ansinew]{inputenc}
\usepackage{a4wide}
\usepackage{color}
\newcommand{\be}{\begin{eqnarray}}
\newcommand{\ee}{\end{eqnarray}}
\newtheorem{theo}{Theorem}

\newtheorem{lemma}{Lemma}

\newtheorem{prop}{Proposition}
\newcommand{\R}{\mathbb R}
\newcommand{\C}{\mathbb C}
\newcommand{\bxi}{\boldsymbol\xi}

\newcommand{\bu}{{\bf u}}

\newcommand{\talpha}{{\tilde\alpha}}

\begin{document}

\title{Decay and Subluminality of Modes of all Wave Numbers in the Relativistic Dynamics of\\
Viscous and Heat Conductive Fluids}

\author{\it Heinrich Freist\"uhler\thanks{Department of Mathematics and Statistics,
University of Konstanz, 78457 Konstanz, Germany; supported by DFG
through Grant FR 822/10-1},
Moritz Reintjes$^*$, {\rm and}\  Blake Temple\thanks{Department of Mathematics,
University of California at Davis, Davis CA 95616, USA.}}

\maketitle

\vskip -4cm
\begin{abstract}
To further confirm the causality and stability of a second-order hyperbolic
system of partial differential equations
that models the relativistic dynamics
of barotropic fluids with viscosity and heat conduction
(H. Freist\"uhler and B. Temple, J. Math.\ Phys.\ 59 (2018)),
this paper studies the
Fourier-Laplace modes of this system
and shows that all such modes, relative to arbitrary Lorentz frames,
(a) decay with increasing time and (b) travel at subluminal speeds. Stability is also shown
for the related model of non-barotropic fluids (H. Freist\"uhler and B. Temple.
Proc.\ R. Soc.\ A 470 (2014) and Proc.\ R. Soc.\ A 473 (2017)).
Even though
these properties had been known for a while in the sense of numerical evidence,
the fully analytical proofs for the subluminality of modes of arbitrary wave numbers
in arbitrary frames given here appear to be the first regarding
\emph{any}\ five-field formulation of dissipative relativistic fluid dynamics.
\end{abstract}

\newpage
\section{Introduction}\label{Introduction}
In this paper, we study a hyperbolic model
for the dynamics of relativistic viscous fluids 
and confirm its stability and causality in terms of decay and subluminal speed of its Fourier-Laplace modes.
The result is complete in the sense that it covers modes of arbitrary wave number and with respect
to arbitrary Lorentz frames.
Concretely, we address the system
\be\label{Teq}
{\partial \over \partial x^\beta}\left(T^{\alpha\beta}
+\Delta T^{\alpha\beta}_\Box
\right)
=0
\ee
which was recently proposed
as a candidate for a relativistic version of Navier-Stokes \cite{FT1,FT3}.
While in \eqref{Teq}
\be \label{T}
T^{\alpha\beta}=\rho U^\alpha U^\beta+p \Pi^{\alpha\beta}
\ee
is the usual ideal-fluid energy-momentum tensor, with $\rho,p, U^\alpha$ the fluid's energy,
pressure and 4-velocity, and $\Pi^{\alpha\beta}=g^{\alpha\beta}+U^\alpha U^\beta$ the projector
on the orthogonal complement of $U^\alpha$, the dissipation tensor reads
\be\label{DeltaT}
\begin{aligned}
-\Delta T^{\alpha\beta}_\Box
&=&\eta \Pi^{\alpha\gamma}\Pi^{\beta\delta}
\left[\frac{\partial U_{\gamma}}{\partial x^{\delta}}
+\frac{\partial U_{\delta}}{\partial x^{\gamma}}
-\frac{2}{3}g_{\gamma\delta}\frac{\partial U^{\epsilon}}{\partial x^{\epsilon}}\right]
+\check\zeta \Pi^{\alpha\beta}\frac{\partial U^{\gamma}}{\partial x^{\gamma}}\\
&&+\sigma\left[U^{\alpha}U^{\beta}\frac{\partial U^{\gamma}}{\partial x^{\gamma}}
-\left(\Pi^{\alpha\gamma}U^{\beta}
+\Pi^{\beta\gamma}U^{\alpha}\right)
U^{\delta}\frac{\partial U_{\gamma}}{\partial x^{\delta}}\right]
\\
&&+
\chi
\bigg[
\bigg(U^\alpha\frac{\partial \theta}{\partial x_\beta}
+U^\beta\frac{\partial \theta}{\partial x_\alpha}\bigg)
-g^{\alpha\beta}U^\gamma\frac{\partial \theta}{\partial x^\gamma}
\bigg],
\end{aligned}
\ee
where 
\be\label{choicetildezeta}
\sigma=((4/3)\eta+\zeta)/(1-c_s^2)-c_s^2\chi\theta
\quad\text{and}\quad
\check\zeta=\zeta+c_s^2\sigma-c_s^2(1-c_s^2)\chi\theta
\ee
with $\eta,\zeta,\chi$ the coefficients of shear viscosity, bulk viscosity, and thermal conductivity,
and $0<c_s<1$ the speed of sound.
We refer the interested reader to \cite{FT3} for a detailed derivation and justification of
equations \eqref{Teq} with \eqref{DeltaT}, \eqref{choicetildezeta},
and to \cite{FT1,FT2,F20} for brief summaries of the history and current developments of
the mathematical theory of relativistic viscous fluid flow.  As in \cite{FT3} we will assume that
the fluid is barotropic, i.e., pressure and sound speed are functions of the
energy,
\be\label{P}
p=P(\rho),\quad c_s^2=P'(\rho),
\ee
which implies that the local state of the fluid has four degrees of freedom (three for the velocity,
which is constrained by $U_\alpha U^\alpha=-1$, and one for the thermodynamics) and correspondingly
the conservation laws \eqref{Teq} for energy-momentum suffice to determine the dynamics.

We focus here on Fourier-Laplace modes. These are solutions of the form
\be\label{FLmode}
\hat\psi_\alpha
e^{\lambda t+i\bxi\cdot \boldsymbol x}
\quad (\lambda,\bxi)\in \C\times\R^3,
\ee
of the linearization of \eqref{Teq} at some homogeneous reference state, with $\boldsymbol x$
referring to the spatial variables $x^1,x^2,x^3$  and $t$ to time $x^0$.
In relativity, the notion of a Fourier-Laplace mode with
its distinction between spatial and temporal frequencies is relative to the choice of a Lorentz
frame.
A non-trivial mode, \eqref{FLmode} with $\bxi\neq 0$, is called
(a) decaying, neutral, or growing if
$\hbox{Re}\{\lambda\}<0,=0,$ or $>0$, respectively, and
(b) subluminal, luminal, or superluminal if $|\hbox{Im}\{\lambda\}|<|\bxi|,=|\bxi|$, or
$>|\bxi|$, respectively.

\begin{theo}\label{ThmD}
For any values
\be\label{etahatzeta}
\eta>0,\quad \check\zeta\ge -\frac13\eta,\quad \sigma=\frac43\eta+\check\zeta,\quad\chi\ge 0,
\ee
the viscous system \eqref{Teq} with \eqref{DeltaT}--\eqref{P}
is dissipative in the sense that all its non-trivial Fourier-Laplace modes with respect to any
Lorentz frame are decaying.
\end{theo}
\begin{theo}\label{ThmC}
For system \eqref{Teq} with \eqref{DeltaT}--\eqref{P} 
all non-trivial Fourier-Laplace modes with respect to any Lorentz frame are subluminal.
\end{theo}

Theorem 1 is proved in Section 2, Theorem 2 in Sections 3 and 4. 
In Section 5, we address a counterpart of Theorem 1 for non-barotropic fluids; the corresponding
model, taken from \cite{FT2}, is introduced there.  

The interpretation of Theorem 1 is clear -- it gives linear stability of homogeneous states in
all reference frames.  Rigorous stability results at the fully nonlinear level have been
accomplished by  Sroczinski \cite{S1,S2}.

Theorem 2  is related to  causality. Regarding modes, causality means that
$\lim_{|\bxi|\to\infty}(|\lambda(\bxi)|/|\bxi|)$ $\le 1$, a property that for
the models considered in the present paper has been checked in \cite{FT2,FT3}.
Our theorem proves a finer property, namely the fact that all \emph{phase velocities},
with $0<|\bxi|<\infty$,  are less
than the speed of light. Recall the recurrent considerations in the literature on interesting physical
systems which do not have this property (see \cite{B} for an early documentation).
For our models, the property is \emph{plausible}, as
both the second-order operator (corresponding to ``$|\bxi|=\infty$'') and the first order operator
(corresponding to ``$|\bxi|= 0$'') have  (sub-/)luminal phase velocity. Our \emph{proof} suggests the possibility of a more general principle granting
such ``interpolation''
between hyperbolic operators of different orders,
but the potential identification of such principle is here left to future work.

\section{Modes and stability in arbitrary frames}
\setcounter{equation}{0}
We first note that a solution to the linearized PDE which is a Fourier-Laplace mode with respect
to a given
frame is generally not a Fourier-Laplace mode with respect to a different frame. We thus briefly
consider \emph{general modes}, i.e., solutions of the linearized
equations of the form
\be\label{mode}
{{\hat\psi_\alpha}}
e^{i\xi^\beta  x_\beta},
\quad \xi^\beta\in\C^{4},
\ee
with, in principle, arbitrary \emph{complex} wave vectors. This \emph{is} a covariant notion;
under frame changes, the amplitude $\hat\psi_\alpha$ and notably the real part $\text{\rm Re }\xi^\beta$
and the imaginary part $\text{\rm Im }\xi^\beta$ of the wave vector transform as 4-vectors.
A general mode \eqref{mode} is a Fourier-Laplace
mode with respect to \emph{some} Lorentz frame if and only if
\be    \label{FLM_invariant_form}
\text{\rm Im }\xi^\beta \text{ is timelike or zero,}
\ee
in the latter case it is \emph{neutral}.
 It is a Fourier-Laplace mode with respect to the fluid rest frame,
if $\text{\rm Im }\xi^\beta$ is a multiple of the future-pointing fluid $4$-velocity
$U^\beta$ of the homogeneous reference state. A Fourier-Laplace mode is decaying if
$\text{\rm Im }\xi^\beta$ is future-pointing, and growing if $\text{\rm Im }\xi^\beta$ is past-pointing.
A Fourier-Laplace mode is \emph{subluminal} or \emph{luminal} if
\be
\text{\rm Re }\xi^\beta \text{ is spacelike or lightlike}.
\ee

We will study solutions that are Fourier-Laplace modes with respect to \emph{some} frame
by looking at them in the fluid's rest frame. To prepare for that perspective,
consider now an arbitrary general mode from the point of view of its
representation in the fluid's rest frame, let this representation be given by
\be\label{genmodeinrestframe}
\hat\psi_\alpha
e^{\lambda t+i\bxi\cdot \boldsymbol x},
\quad (\lambda,\bxi)\in \C\times\C^3;
\ee
i.e., relative to \eqref{mode} we have written $\bxi$ for $(\xi^1,\xi^2,\xi^3)$ and $\lambda$
for $i\xi^0$;
 the only difference of \eqref{genmodeinrestframe} from \eqref{FLmode} is that
$\bxi$ now is complex.
Assume that with respect to a fixed ``rest frame``, the dispersion relation for a covariant linear
system of partial differential equations is given by
\be\label{drrf}
0=\Pi(\lambda,\bxi)
\ee
i.e., $\Pi$ results from the symbol of the differential operator describing the system
with respect to the rest frame.

Consider now a solution to the PDE system
which is a Fourier-Laplace mode
\be\label{tildemode}
\widetilde{\hat\psi_\alpha}
e^{\tilde\lambda \tilde t+i\tilde\bxi\cdot\tilde{\boldsymbol x}}
=
\widetilde{{\hat\psi}_\alpha}
e^{Re\{\tilde\lambda\}\tilde t}e^{i\left(Im\{\tilde\lambda\}t+\tilde\bxi\cdot \tilde{\boldsymbol x} \right)},
\quad (\tilde\lambda,\tilde\bxi)\in \C\times\R^3,
\ee
with respect to a frame which is not the rest frame, but differs from it by a Lorentz boost.
In order to nevertheless study it by means
of the rest-frame dispersion relation \eqref{drrf} we note (cf.\ also \cite{HL}) that
if \eqref{FLmode} represents the same mode with respect to the rest frame,  the complex wave
vectors
$(\tilde\lambda,i\tilde\bxi)$ and $(\lambda,i\bxi)$ are related
through the Lorentz transformation
\be
(\lambda,\bxi)=\Lambda_v(\tilde\lambda,\tilde\bxi),
\ee
from the non-rest frame to the rest frame, i.e.,
\be\label{Lorentz_var}
\begin{aligned}
\lambda&=\gamma\tilde\lambda+i\gamma v\tilde\xi_1 \\
\xi_1&=-i\gamma v\tilde\lambda+\gamma\tilde\xi_1 \\
\xi_2&=\tilde\xi_2\\
\xi_3&=\tilde\xi_3,
\end{aligned}
\ee
with some $$v\in(-1,1)\text{ and }\gamma=(1-v^2)^{-1/2}.$$
The Fourier-Laplace modes in the tilde 
frame are thus controlled by the dispersion relation
\be \label{disp_rel}
0=\tilde\Pi_v(\tilde\lambda,\tilde\bxi)
\ee
with
$$
\tilde\Pi_v(\tilde\lambda,\tilde\bxi)
\equiv
\Pi(\Lambda_v(\tilde\lambda,\tilde\bxi))=\Pi(\lambda,\bxi).
$$

In our specific situation, the linearized equations of motion in the fluid's rest frame,
\be 
\begin{aligned}\label{PDElin}
\frac{\rho+p}\theta\left\{c_s^{-2}
{\partial \theta\over \partial t}
+\theta\nabla\cdot\bu
\right\}
+\left\{\chi\left(\frac{\partial^2\theta}{\partial t^2}-\Delta\theta\right)\right\}
&=0,\\
\frac{\rho+p}\theta\left\{
\theta{\partial \bu\over \partial t}
+\nabla\theta
\right\}
+\left\{{\sigma}\frac{\partial^2\bu}{\partial t^2}
-({\eta}\,\nabla\cdot{\cal S}\bu+{\check\zeta}\,\nabla(\nabla\cdot \bu))\right\}&=0
\end{aligned}
\ee
yield, after a simple scaling, the dispersion relation
\be\label{dr}
0=\Pi(\lambda,\bxi)\equiv \det M(\lambda,\bxi),
\ee
with
(cf.\ \cite{FT3})
$$
M(\lambda,\bxi)
=
\begin{pmatrix}
c_s^{-2}\lambda  & i\bxi^\top \\
i\bxi & \lambda I
\end{pmatrix}
+
\begin{pmatrix}
\chi(\lambda^2+\bxi^2)
   & 0 \\
0&  N(\lambda,\bxi)
\end{pmatrix},
$$
where
$$
N(\lambda,\bxi)=\left(\sigma\lambda^2+\eta\bxi^2\right)I
+\left(\check\zeta+{1\over 3}\eta\right)\bxi\bxi^\top,
$$
with $\bxi^2\equiv\bxi^\top\bxi=\xi_1^2+\xi_2^2+\xi_3^2,$
and the dispersion relation factors as
\be \label{drsplit}
\Pi(\lambda,\bxi)=\Pi^L(\lambda,\bxi) \big(\Pi^T(\lambda,\bxi) \big)^2
\ee
with
$$
\Pi^L(\lambda,\bxi)=\{c_s^{-2}\lambda +\chi (\lambda^2+\bxi^2)\}
\{\lambda + \sigma\left(\lambda^2+\bxi^2\right)\}
+ \bxi^2
$$
and
$$
\Pi^T(\lambda,\bxi)=\lambda+\sigma\lambda^2+\eta\bxi^2;
$$
to simplify notation, we will use the again scaled versions
\be
\Pi^L(\lambda,\bxi)=\{\lambda +\hat\chi (\lambda^2+\bxi^2)\}
\{\lambda +\left(\lambda^2+\bxi^2\right)\}
+ c_s^2 \bxi^2
\ee
and
\be
\Pi^T(\lambda,\bxi)
=\lambda+\hat\sigma\lambda^2+\bxi^2
\ee
with $\hat\chi=c_s^2\chi/\sigma\ge 0$ and $\hat\sigma=\sigma/\eta\ge 1$.

Correspondingly, the Lorentz transformed dispersion relation
\eqref{disp_rel} also factors, as
\be\label{drtildesplit}
\tilde\Pi_v(\tilde\lambda,\tilde\bxi) =
\tilde\Pi^L_v(\tilde\lambda,\tilde\bxi)  \big(\tilde\Pi^T_v(\tilde\lambda,\tilde\bxi)\big)^2,
\ee
with
\be \label{tildePiLtildePiT}
\tilde\Pi^L_v(\tilde\lambda,\tilde\bxi)
\equiv
\Pi^L(\lambda,\bxi),\quad
\tilde\Pi^T_v(\tilde\lambda,\tilde\bxi)
\equiv
\Pi^T(\lambda,\bxi).
\ee
Concretely, setting $r^2\equiv \xi_2^2 +\xi_3^2 = \tilde\xi_2^2 + \tilde\xi_3^2$
and substituting the identities
$$
\lambda^2 + \bxi^2 = \tilde\lambda^2 + \tilde\bxi^2
\hspace{1cm} \text{and} \hspace{1cm}
\bxi^2 = r^2 - \gamma^2 (\tilde\lambda v + i \tilde\xi_1)^2,
$$
into $\tilde\Pi^L$, we obtain
\begin{eqnarray}  \label{Pi_L_gen2}
\tilde\Pi^L(\tilde\lambda,\tilde\bxi)
& = & \hat\chi  \big(\tilde\lambda^2+\tilde\bxi^2\big)^2 + \gamma (\hat\chi +1) \big(\tilde\lambda + i v \tilde\xi_1\big)\big(\tilde\lambda^2+\tilde\bxi^2\big) + \gamma^2 \big(\tilde\lambda + i v \tilde\xi_1\big)^2  \cr  & &    +\; c_s^2\big( r^2 - \gamma^2 (\tilde\lambda v + i \tilde\xi_1)^2 \big),
\end{eqnarray}
and a straightforward computation gives
\begin{align} \label{Pi_T_gen2}
\tilde\Pi^T(\tilde\lambda,\tilde\bxi)
&=\gamma(\tilde\lambda+iv\tilde\xi_1) +\hat\sigma(\gamma\tilde\lambda+i\gamma v\tilde\xi_1)^2 + (-i\gamma v\tilde\lambda+\gamma\tilde\xi_1)^2+r^2.
\end{align}
We now turn to the proof of Theorem 1, which expresses  a general principle
to conclude stability in all frames from stability in one frame, based on the idea
to bound Re$(\lambda)$ away from 0 by showing that neutral modes do not exist.\footnote{This idea
seems to have been brought up first in \cite{FT1}, Lemmas 6.4 and 6.6.}
\begin{prop}
There does not exist a nontrivial Fourier-Laplace mode of \eqref{PDElin}
in any Lorentz frame, which is neutral.
\end{prop}
\begin{proof}
While the question of whether a given \emph{general} mode \eqref{mode} is a Fourier-Laplace mode
with respect to a Lorentz frame generically depends on the choice of the frame,
a  mode that is a \emph{neutral} Fourier-Laplace mode with respect to one Lorentz frame, obviously
is a neutral Fourier-Laplace mode with respect to any other Lorentz frame as well. Thus if a
neutral nontrivial Fourier-Laplace mode with respect to some frame existed, it would readily be
a neutral nontrivial Fourier-Laplace mode with respect to the fluid's rest frame. But it is
known from \cite{FT3} that with respect to the fluid's rest frame, the model under consideration
allows only decaying, thus no neutral nontrivial Fourier-Laplace modes.
\end{proof}

{ \it Proof of Theorem 1.}
Assume that with respect to some frame, there is a non-trivial
Fourier-Laplace mode that does not decay, briefly: for some $\tilde v\in(-1,1)\setminus\{0\}$
there exists
$\tilde\bxi\in\R^3\setminus\{0\}$ such that $\tilde\Pi_{\tilde v}(.,\tilde\bxi)$ has a zero
$\tilde\lambda$ with positive real part. However, the set of (complex) zeroes of a continuously
parameterized family of (complex) polynomials depends continuously on the parameter. Using this
for the family of polynomials $\tilde \Pi_v(.,\tilde\bxi)$ with $\tilde\bxi$ fixed and the
parameter $v$ varying down or up to $0$, noting that during that variation the following
Lemma 1 keeps the set of zeroes within a fixed ball of finite radius $R(|\tilde v|,\{\tilde\bxi\})$ -- so none
can escape or come in across the boundary of that ball --, and recalling that all roots of
$\tilde\Pi_0(.,\tilde\bxi)$ have strictly negative real part, we see that the
Intermediate Value Theorem implies that there exists a value $v^*\in (-1,1)$ for which
$\tilde\Pi_{v^*}(.,\tilde\bxi)$ has a root with vanishing real part. This would contradict
Proposition 1.  \hfill $\Box$

\begin{lemma}
For all $\bar v\in [0,1)$ and any compact set $X\subset \R^3$, there exists $R=R(X,\bar v)>0$ such
that if $\tilde\lambda\in\C$ satisfies      
$
{\tilde\Pi}_v(\tilde\lambda,\tilde\bxi)=0$ with $|v|\le \bar v\text{ and }\tilde\bxi\in X,
$
then
$|\tilde\lambda|\le R.$ In particular, the set of zeros of $\tilde\Pi_v(.,\bxi)$ depends continuously
on the parameter $v$.
\end{lemma}
\begin{proof}
Fix $\bar v$ and $X$ as stated, and consider the family
$$
\tilde\Pi_v(.,\tilde\bxi),\quad
(v,\tilde\bxi)\in\bar X=[-\bar v, \bar v]\times X
$$
of polynomials in $\tilde\lambda$. The assertion is true if (a) all coefficients of these polynomials
satisfy a uniform bound from above and (b) the leading-order coefficient is bounded uniformly away
from $0$. But (a) is clear due to the compactness of $\bar X$; notably, the Lorentz factor $\gamma$ is
bounded by its value $(1-\bar v)^{-1/2}$ for $v=\bar v$. Property (b) comes from the
(though sharp, cf.\ \cite{FT1})
 causality at
$|\bxi|=\infty$: for the
polynomial $\Pi$ in $\lambda$ and $\xi_1$ that gives the
dispersion relation with respect to the rest frame (considering $\xi_2,\xi_3$ as
parameters), the coefficient of the highest power $\lambda^n$ of $\lambda$
has absolute value $a>0$ larger or equal to that of the highest power $\xi_1^n$ of $\xi_1$.  
In the mixing of $\tilde\lambda$ and $\tilde\xi_1$ through the Lorentz transformation
\eqref{Lorentz_var}, the coefficient of $\tilde\lambda^n$ in $\tilde\Pi_v$ thus remains bounded
from below by $(1-\bar v^n)a>0$.
\end{proof}

\section{Subluminality in absence of heat conduction}

In this and the next section we show that our system does not allow for
luminal modes. As we know from \cite{FT3} that for $v=0$ all modes are subluminal,
the continuity of the solution set of \eqref{disp_rel} with respect to the parameter $v$
according to Lemma 1 implies that all modes are subluminal, for any $v\in (-1,1)$  -- the assertion of Theorem 2.

Details of our proof  depend on whether the fluid does allow for heat conduction
or not. In this section, we consider the easier case that it does not, i.e., we assume that
the heat conductivity vanishes, $\chi=0$. We prove the non-existence
of luminal modes in the longitudinal case and in the transverse case separately in two subsections.

\subsection{Longitudinal modes}   \label{Sec_long_modes_chi0}
\setcounter{equation}0
The purpose of this subsection is to prove that there are no longitudinal luminal modes:
\begin{prop}\label{Prop_nonlum_L}
For any $v\in(-1,1)$,
$\tilde\Pi^L_v(\alpha\pm i|\tilde\bxi|,\tilde\bxi)\neq 0$ for all
$\alpha<0$ and $\tilde\bxi\in\R^3\setminus\{0\}$.
\end{prop}

To that end we fix
\be\label{fixcv}
c_s^2\in (0,1)\quad\text{and}\quad v\in(-1,1)
\ee
and lead the assumption that for some $\tilde\bxi\in\R^3\setminus\{0\}$
and some $\tilde\alpha<0$
\be\label{widerspruchsannahme}
\tilde\Pi^L_v(\tilde\alpha+iw,\tilde\bxi)=0\quad\text{with }w=|\tilde\bxi|\text{ or }w=-|\tilde\bxi|
\ee
to a contradiction. The idea is to show that on the one hand the imaginary part of equation \eqref{widerspruchsannahme} forces $\tilde \alpha$ to lie in a certain interval $\mathcal A\subset\R$ on which, on the other hand, the real part of \eqref{widerspruchsannahme} cannot hold.

Writing
$$
a=\gamma^2(1-c_s^2),\quad b= \gamma^2(1-c_s^2v^2),
$$
and using $\gamma^2 -1 = \gamma^2 v^2$ and $\gamma^2 (c_s^2-v^2) - c_s^2 = -a v^2$,
we find from \eqref{Pi_L_gen2} that
\begin{align}
{\rm Im}(\tilde\Pi^L_v(\alpha+iw,\tilde\bxi))
&=   \alpha  \big(2a +  \gamma \alpha \big) v\tilde\xi_1 + \alpha w \big( 2b + 3\gamma\alpha \big) \label{IM_PiL} \\
{\rm Re}(\tilde\Pi^L_v(\alpha+iw,\tilde\bxi)) &=
-av^2 \tilde\xi_1^2
- 2 w \big( a + \gamma \alpha \big) v\tilde\xi_1 + b (\alpha^2- w^2) 
+ \gamma \alpha (\alpha^2 - 2 w^2)  + c_s^2 w^2.      \label{RE_PiL}
\end{align}
Introducing
\begin{equation}  \label{0chi_p1&2}
p_1(\alpha)= 2 a + \gamma\alpha
\hspace{1cm} \text{and} \hspace{1cm}
p_2(\alpha)= 2 b + 3\gamma\alpha,
\end{equation}
we write the imaginary part of \eqref{widerspruchsannahme} as
\begin{eqnarray}  \label{0chi_xi_soln}
p_1(\tilde\alpha)v\tilde\xi_1 = - p_2(\tilde\alpha) w.
\end{eqnarray}
\begin{lemma} \label{0chi_adm-reg_lemma}
(a) For any solution of \eqref{widerspruchsannahme}, $\tilde\alpha$ lies in the
interval
$$
\mathcal A=[\min\{\alpha_-,\alpha_+\},\max\{\alpha_-,\alpha_+\}]
$$
with
\be \label{0chi_alpha_+-}
\alpha_\pm  
= \ -2\gamma \frac{1\pm v}{3\pm v}(1\mp vc_s^2).
\ee
(b) $\mathcal{A}$ consists of a single point if and only if
either $v=0$ or
\be \label{0chi_cr}
c_s^2 = 2/(3-v^2),
\ee
and in this case
\be\label{3alphas}
\alpha_- = \alpha_+ = \frac{-2}{\gamma(3-v^2)} \equiv \alpha_*.
\ee
\end{lemma}

\begin{proof}
To prove part (a), observe that \eqref{0chi_xi_soln} and the inequality $0<|\tilde\xi_1| \leq |w|$ implies
\be  \label{0chi_admis_reg}
|p_2(\tilde\alpha)|\le |vp_1(\tilde\alpha)|.
\ee
Inequality \eqref{0chi_admis_reg} holds if and only if $\tilde\alpha$ lies between the zeros of the affine functions
$$
p_+ \equiv p_2 + v p_1\quad \text{and}\quad p_- \equiv p_2 - v p_1
$$
given by
\begin{equation}  \label{0chi_i+-}
p_\pm(\alpha)   =  (3 \pm v) \gamma\alpha + 2 (b\pm va ),
\end{equation}
and these zeros are $\alpha_+$ and $\alpha_-$.
Part (b) directly follows by equating $\alpha_+$ and $\alpha_-$, which implies that either \eqref{0chi_cr} holds or $v=0$.
\end{proof}

Next, multiplying the real part of the dispersion relation \eqref{widerspruchsannahme} by $p_1^2(\alpha)$
and using \eqref{0chi_xi_soln} to substitute for $v \tilde\xi_1 p_1(\tilde\alpha)$ gives
\begin{equation}   \label{0chi_A-B-eqn}
0= A(\tilde\alpha) w^2 + B(\tilde\alpha)p_1^2(\tilde\alpha),
\end{equation}
with
\begin{align} \label{0chi_A&B}
A(\alpha) &\equiv   -a p_2(\alpha)^2 + 2 \big( a + \gamma \alpha \big) p_1(\alpha)p_2(\alpha)
-  (2 \gamma \alpha  + a) p_1(\alpha)^2 ,  \cr
&= 4\gamma^3\alpha^3 + 4b \gamma^2\alpha^2 - 4a c_s^4   \cr
&=  4 \gamma^2 B(\alpha)   - 4a c_s^4    \\
B(\alpha) &\equiv    \big(  b  + \gamma\alpha \big) \alpha^2 ,
\end{align}
as an equation equivalent to the dispersion relation \eqref{widerspruchsannahme}.
The essential idea is now to control the signs of $A$ and $B$ in a way that shows that
\eqref{0chi_A-B-eqn} cannot hold.
This is the subject of the next two lemmas.

\begin{lemma} \label{0chi_lemma_A}
(i) The polynomial $A$ is positive on the interior of $\mathcal{A}$.\\
(ii) $A$ vanishes on $\partial\mathcal A$ if and only if $c_s^2=2/(3-v^2)$, (in which case $\mathcal{A}$ is a singleton).
\end{lemma}

\begin{proof}
We write $\tilde A(\alpha) \equiv  4v^2  A(\alpha)$ as the quadratic form
\begin{eqnarray} \label{0chi_tA}
\tilde A =   r_1 {p_+}^2     +  2 r_2 p_+ p_-    + r_3  {p_-}^2   ,
\end{eqnarray}
in $p_+,p_-$ with coefficients        
\begin{eqnarray} \nonumber
r_1(\alpha) &=&   - (1-v) \big((1-v) a + 2\gamma \alpha\big)      \cr
r_2(\alpha)&=&      (1-v^2) a  +2\gamma \alpha  \cr
r_3(\alpha) &=&    - (1+v) \big((1+v) a + 2\gamma \alpha\big)   ,
\end{eqnarray}
by substituting $p_2= \frac12(p_+ + p_-)$ and $vp_1=\frac12(p_+ - p_-)$  into \eqref{0chi_A&B}, and using the identities $v^2 \gamma^2 = \gamma^2 -1$, $\gamma^{-2}=1-v^2$.

We now investigate whether $\tilde A$ vanishes at a boundary point of $\mathcal A$. It clearly does vanish wherever both $p_+$ and $p_-$ vanish simultaneously; this is only the case when $c_s^2=2/(3-v^2)$, at the point displayed in equation \eqref{3alphas}. Assume now that
\be
c_s^2\neq \frac{2}{3-v^2} \quad \text{and} \quad  v \neq 0.
\ee
We first note that neither $p_+$ and $r_3$ nor $p_-$ and $r_1$ have common zeros. For assuming, say,
that $p_+$ and $r_3$ did, i.e.,
\begin{equation}
\alpha_+ \equiv  -2\gamma \frac{1+v}{3+v}(1-vc_s^2) \ = \
\alpha_3 \equiv  \frac{c_s^2 - 1}{2 \gamma (1-v)}     ,
\end{equation}
or equivalently
\begin{eqnarray}
(3+v)(c_s^2 - 1)   =   -4 \gamma^2 (1- v^2)(1- v c_s^2) ,
\end{eqnarray}
yields the contradiction that $c_s^2 = -\frac13$.
It now follows from \eqref{0chi_tA} that $\tilde A$ does not vanish at the boundary points of $\mathcal
A$ unless they coincide. In the special case $v=0$,
when by \eqref{3alphas} $\mathcal{A}$ collapses to the single point $\alpha_\pm = - 2/3$,
we find that $A(-2/3)= 16/27 - 4ac_s^4$, which vanishes if and only if $c_s^2 = 2/3$,
i.e., when \eqref{0chi_cr} holds. This proves assertion (ii).

To show (i), note that the third order polynomial $A$  satisfies
\be
A(-c_s^2/\gamma)=0,\quad A(0)<0,\quad\text{and} \quad A(+\infty)>0
\ee
and thus has both a negative and a positive root, and all roots of $A$ are real. By this, since $A$ has no zeros on $\partial A$ for
\be\label{notsep}
c_s^2\neq \frac{2}{3-v^2},
\ee
it follows that $A$ is positive on $\mathcal A$ for \emph{all}\ values of the parameters
\eqref{fixcv} as soon as it is positive on  $\mathcal A$ for one pair $(c_s^2,v)$ with
$c_s^2<2/(3-v^2)$ as well as for one pair with $c_s^2>2/(3-v^2)$. Indeed, the limiting cases
$(0,0)$ and $(1,0)$ have this property.   To see this, recall that  for $v=0$,  $\mathcal A$ is the singleton $\{-2/3\}$ and
\begin{equation}
A(\alpha) =  4 \big( \alpha^3 +  \alpha^2 - (1-c_s^2) c_s^4 \big),
\end{equation}
and thus on $\mathcal A$, $A$ assumes the value
\begin{equation}
{A}(-2/3)  =   \frac{16}{27} - 4(1-c_s^2) c_s^4  =  \frac{16}{27}>0
\end{equation}
both for $c_s^2=0$ and $c_s^2=1$. This completes the proof.
\end{proof}

\begin{lemma}  \label{0chi_lemma_B}
The polynomial $B$ is positive on $\mathcal{A}$.
\end{lemma}
\begin{proof}
By \eqref{0chi_A&B}, $B$ differs from $A$ by addition with a positive constant,  $4\gamma^2 B(\alpha)   = A(\alpha) + 4a c_s^4$. Thus, since by Lemma \ref{0chi_lemma_A} $A$ is positive on $\mathcal{A}$, $B$ is positive on $\mathcal{A}$ as well.
\end{proof}

In view of equation \eqref{0chi_A-B-eqn}, Lemmas \ref{0chi_lemma_A} and \ref{0chi_lemma_B}
readily imply the assertion of Proposition 2 except in the case \eqref{0chi_cr}.
In this case, in which $\alpha_+ = \alpha_- = \alpha_*$ by Lemma \ref{0chi_adm-reg_lemma},
we find from \eqref{0chi_p1&2}
that $\alpha_*$ is also a common zero of $p_1$ and $p_2$.
This implies by \eqref{0chi_xi_soln} that the imaginary part of the dispersion relation
\eqref{widerspruchsannahme} vanishes for $\talpha = \alpha_*$ and we cannot exactly
argue as above.

\begin{lemma}\label{sepcaseforchi=0}
In case $c_s^2=2/(3-v^2)$, one finds ${\rm Re}(\tilde\Pi_v^L(\alpha_*+iw,\tilde\bxi))\neq 0$.
\end{lemma}

\begin{proof}
In this case, $p_1,p_2$ and $A$ have the common zero $\alpha_*$ and $\hat A\equiv A/p_1^2$
satisfies $\hat A(\alpha_*)>0$. Instead of \eqref{0chi_A-B-eqn} one uses
$$
0=\hat A(\alpha_*)w^2+B(\alpha_*)
$$
which is a contradiction by Lemma \ref{0chi_lemma_B}.   
\end{proof}

\subsection{Transverse modes}
The purpose of this subsection is to prove that there are no transverse luminal modes:       
\begin{prop}\label{Prop_nonlum_T}
$\tilde\Pi^T_v(\alpha\pm i|\tilde\bxi|,\tilde\bxi)\neq 0$ for all
$\alpha<0$ and $\tilde\bxi\in\R^3\setminus\{0\}$.
\end{prop}

\begin{proof}
We fix $c_s^2\in (0,1)\quad\text{and}\quad v\in(-1,1)$  and lead the assumption that for some $\tilde\bxi\in\R^3\setminus\{0\}$ and some $\tilde\alpha<0$
\be\label{wid_ann_T}
\tilde\Pi_v^T(\tilde\alpha+iw,\tilde\bxi)=0\quad\text{with} \quad w=|\tilde\bxi|\text{ or }w=-|\tilde\bxi|
\ee
to a contradiction, following the idea of proof for Proposition \ref{Prop_nonlum_L} above. That is, in a first step we show the imaginary part of equation \eqref{wid_ann_T} forces $\tilde \alpha$ to lie again in some range $\mathcal A\subset\R$, and in the second step we show that the real part of \eqref{wid_ann_T} cannot hold on $\mathcal A$.  To begin, we substitute the ansatz $\tilde\lambda = \talpha + i w$ into the expression for $\tilde\Pi^T$ in \eqref{Pi_T_gen2}, which gives
\begin{align} \nonumber
\tilde\Pi_v^T(\tilde\lambda,\tilde\bxi)
&= \gamma(\talpha+i(w+v\tilde\xi_1))+\hat\sigma\gamma^2\big(\talpha+i(w+v\tilde\xi_1)\big)^2    +   \gamma^2 \big(-i v\talpha+(vw+\tilde\xi_1)\big)^2+ w^2 - \tilde\xi_1^2.
\end{align}
From this we find that the real and imaginary parts of condition \eqref{wid_ann_T} are given by
\begin{align}
\label{ImPiT}
0&=\text{Im\ \!}\tilde\Pi_v^T (\tilde\alpha+iw,\tilde\bxi)=\gamma(w+v\tilde\xi_1)
+2\hat\sigma\gamma^2\talpha(w+v\tilde\xi_1)
-2\gamma^2 v\talpha(vw+\tilde\xi_1),\\
\label{RePiT}
0&=\text{Re\ \!}\tilde\Pi_v^T (\tilde\alpha+iw,\tilde\bxi)
 = \big( 1 + (\hat\sigma-v^2)\gamma\talpha \big) \gamma \talpha + \gamma^2 (1-\hat\sigma)\big(w^2 + 2 wv \tilde\xi_1 + v^2 \tilde\xi_1^2 \big).
\end{align}
For the first step, we write \eqref{ImPiT} as
\begin{equation}\label{k}
 0=q_1(\talpha) v \tilde\xi_1+q_2(\talpha)w
\end{equation}
for
\begin{eqnarray}
q_1(\alpha) \equiv  1+2\gamma\alpha(\hat\sigma-1)
\hspace{1cm} \text{and} \hspace{1cm}
q_2(\alpha) \equiv   1+2\gamma\alpha(\hat\sigma-v^2).
\end{eqnarray}
Since $|\tilde\xi_1| \leq |w|$, equation \eqref{k} implies that
\be
\talpha \in   \mathcal A \equiv \{\alpha:\: |q_2(\alpha)|\le |vq_1(\alpha)|\}
\ee
with $\mathcal A$ the interval between the
zeros of
$$
q_\pm=q_2\pm vq_1= \big(2\gamma (\hat\sigma\mp v)\alpha +1 \big)(1\pm v),
$$
i.e.,
$$
\mathcal A=[\min\{\alpha_-\alpha_+\},\max\{\alpha_-\alpha_+\}]\quad\text{with }
\alpha_\pm = -\frac{1}{2\gamma (\hat\sigma \mp v)} <0.
$$

For the second step, we multiply \eqref{RePiT} by
$q_1^2$ and, after substitution of \eqref{k}, write the resulting expression as
\begin{equation}\label{Aw^2+B=0}
0=A(\talpha)w^2+B(\talpha)
\end{equation}
with                
\begin{eqnarray} \nonumber
A(\alpha) &=& -4 (\hat\sigma-1) \alpha^2  \quad   \leq 0, \quad \text{as } \hat\sigma \geq 1,     \cr
B(\alpha)  &=& q_1(\alpha)^2 \big(1 + (\hat\sigma -v^2)\gamma \alpha \big)\gamma \alpha .
\end{eqnarray}
As the zeros $0$, $\alpha_0$ and $\alpha_1$ of $B$ lie outside of $\mathcal{A}$,
$$
\alpha_0 = -\frac{1}{\gamma (\hat\sigma-v^2)} < \alpha_\pm < \alpha_1 = -\frac{1}{2\gamma(\hat\sigma-1)} <0,
$$ 
equation \eqref{Aw^2+B=0} is violated on $\mathcal A$, implying the sought after contradiction.
\end{proof}

\section{ Subluminality in presence of heat conduction}
\setcounter{equation}0
We turn to fluids allowing for heat conduction, $\chi>0$.
As heat conduction does not influence linear transverse
waves -- formally: $\chi$ does not appear in $\Pi_v^T$ --, it is obvious that  Proposition
\ref{Prop_nonlum_T} keeps valid.
The proof of Theorem \ref{ThmC} will thus be completed once we have shown
\begin{prop} \label{nonlum_L_chi>0}
The assertion of Proposition  \ref{Prop_nonlum_L} remains true for $\chi>0$.
\end{prop}
The rest of the section is devoted to proving this. In principle we proceed as in section 3.1. 
We begin by fixing the parameters
$$
c_s^2 \in (0,1),  \hspace{1cm} \hat\chi = c_s^2\chi/\sigma >0  \hspace{1cm} \text{and} \hspace{1cm} v\in (-1,1),
$$
assume for contradiction there exist $\talpha< 0$ and $\tilde\bxi \in \R^3\setminus \{0\}$
such that \eqref{widerspruchsannahme} holds, and show analogues of Lemmas \ref{0chi_adm-reg_lemma}
to \ref{sepcaseforchi=0}.
This time,
\begin{align}
{\rm Im}(\tilde\Pi_v^L(\alpha+iw,\bxi))
&=  4 \hat\chi  w  \alpha^3  + \gamma \alpha^2 (\hat\chi +1) \big( 3 w+ v \tilde\xi_1  \big)  + 2 \gamma^2 \alpha (w+ v \tilde\xi_1)  -     2 c_s^2  \gamma^2  \alpha v (wv+ \tilde\xi_1)   \label{Im_Pi}   \\
{\rm Re}(\tilde\Pi_v^L(\alpha+iw,\bxi)))
&=   \hat\chi (\alpha^2  -4 w^2)\alpha^2
+ \gamma \alpha (\hat\chi +1) \big(\alpha^2 -2w(w+ v \tilde\xi_1)\big)
+ \gamma^2 \big(\alpha^2 - (w+ v \tilde\xi_1)^2\big)  \cr
& \ + c_s^2\big\{ r^2 - \gamma^2 \big(\alpha^2 v^2 - (wv+ \tilde\xi_1)^2\big) \big\}.    \label{Re_Pi}
\end{align}
Note first that equation \eqref{Im_Pi} is equivalent to
\begin{equation}  \label{xi_general_chi}
v p_1(\talpha) \tilde\xi_1 = - w  p_2(\talpha)
\end{equation}
with now
\begin{equation} \label{p1_&_p2}
p_1(\alpha) = (\hat\chi+1)\gamma \alpha + 2a
\hspace{1cm} \text{and} \hspace{1cm}
p_2(\alpha) =  4 \hat\chi\alpha^2 + 3(\hat\chi+1)\gamma \alpha + 2b,
\end{equation}
where $a = \gamma^2(1-c_s^2)$ and $b  = \gamma^2(1-c_s^2v^2)$ as before. The constraint
$w^2\ge\tilde\xi_1^2$ again implies that $\talpha$ must lie in the set
\be \label{admis_reg}
\mathcal{A} \equiv \{\alpha:|p_2(\alpha)|\le|vp_1(\alpha)|\},
\ee
whose boundary points are the roots of the polynomials
$p_\pm \equiv p_2 \pm v p_1$ given by
\begin{equation} \label{p+_&_p-}
p_\pm(\alpha) =4 \hat\chi\alpha^2 + (3\pm v)(\hat\chi+1)\gamma \alpha + 2(b\pm va).
\end{equation}

\begin{lemma}  \label{Lemma_admissibility-region}
\textbf{(i)} The two polynomials $p_+$ and $p_-$ each have two distinct zeros,
$\alpha_\pm^- < \alpha_\pm^+$. They have a common zero if and only if either $v=0$,
in which case they
have both zeros in common,  $\alpha_-^+=\alpha_+^+$ and $\alpha_-^-=\alpha _+^-$ , or        
\begin{eqnarray} \label{cr}
0 = \Gamma(c_s^2,\hat\chi,v)   \equiv   8 \hat\chi a^2 + \gamma^2 (b-3a)(\hat\chi+1)^2,
\end{eqnarray}
in which case $\alpha_-^+=\alpha_+^+$, while $\alpha_-^-\neq\alpha _+^-$.  \\
\textbf{(ii)} $\mathcal{A} = \mathcal{A}^- \cup \mathcal{A}^+$ with two disjoint compact intervals
$\mathcal A^- < \mathcal A^+$ having $\partial \mathcal{A^\pm} = \{\alpha_-^\pm,\alpha_+^\pm\}$. \\
\textbf{(iii)}  If $v=0$, the two intervals $\mathcal A^-$ and $\mathcal A^+$ are singletons, i.e., as
$v\rightarrow 0$ each interval $\mathcal A^\pm$ collapses continuously to a single point.
These points are the two distinct zeros of $p_+=p_-$. \\
\textbf{(iv)} For $v\neq 0$, $\mathcal A^-$ always has non-empty interior, while $\mathcal A^+$
collapses to a single point if and only if \eqref{cr} holds.
\end{lemma}

\begin{proof}
Part (i) directly follows from the discriminant $\Delta$ of $p_\pm$,
\begin{equation} \nonumber
\Delta= \gamma^2(3\pm v)^2(\hat\chi+1)^2-32\hat\chi(b\pm va),
\end{equation}
being strictly positive. This can be seen by using $(\hat\chi+1)^2/{\hat\chi}\geq 4$
to estimate
$$
\frac{b\pm va }{\gamma^2 (3\pm v)^2}<\frac18\le \frac{(\hat\chi+1)^2}{32\hat\chi}.
$$
Thus $p^+$ and $p^-$ both have two distinct zeros.
Clearly, if $v=0$, then $p_+=p_-=p_2$ and thus $\alpha_-^+=\alpha_+^+$ and $\alpha_-^-=\alpha _+^-$.
Assume now that $v\neq 0$ and $p^+$ and $p^-$, thus equivalently $p_1$ and $p_2$,
have a common zero. But the only zero of $p_1$ is
\begin{equation}    \label{zero_p1}
\alpha_* \equiv -\frac{2a}{\gamma(\hat\chi+1)} =\alpha^+_-=\alpha^+_+,
\quad \quad \text{and } \quad
p_2(\alpha_*) = 2\frac{(\hat\chi+1)^2}{\gamma^6} \Gamma(c_s^2,\hat\chi,v).
\end{equation}
Thus for $v\neq 0$, $\alpha_-^+=\alpha_+^+$ if and only if
\eqref{cr} holds while always $\alpha_-^-\neq\alpha_+^-$, and in any case
\be\label{4alphas}
\alpha_-^-,\alpha_+^-<\alpha_-^+,\alpha_+^+.
\ee

To see part (ii), note first that
as $|p_2(\alpha)| > |vp_1(\alpha)|$ for sufficiently large $|\alpha|$,
$\mathcal{A}$ is bounded.
The assertion thus follows from the fact $\partial \mathcal A=\{\alpha_-^-,\alpha_+^-,\alpha_-^+,\alpha_+^+\}$
and the above inequalities and equalities holding between these four numbers.

Parts (iii) and (iv) are direct consequences of (i) and (ii).
\end{proof}

Using that the imaginary part \eqref{Im_Pi} of the dispersion relation
vanishes for $\alpha=\tilde\alpha$,
we evaluate its real part \eqref{Re_Pi}
as
\begin{equation} \label{Aw^2+B}
A(\talpha) w^2 + B(\talpha) p_1^2(\talpha)   =0,
\end{equation}
in terms of the polynomials
\begin{align}
A(\alpha) &\equiv  \ -  a\: p_2(\alpha)^2  + 2 \big(\gamma \alpha (\hat\chi +1) +a \big) \: p_2(\alpha) p_1(\alpha)  - \big(4 \hat\chi \alpha^2  + 2 \gamma \alpha (\hat\chi +1) + a\big) p_1(\alpha)^2   \cr
& = \ 4\gamma^2 B(\alpha) - k    \label{A} \\
B(\alpha) &\equiv  \  \big(\hat\chi \alpha^2  + (\hat\chi +1) \gamma\alpha + b \big) \alpha^2 \label{B} 
\end{align}
where $k \equiv 4 a c_s^4 \big( (\hat\chi-1)^2 + 4\hat\chi c_s^2 \big)^{-1}>0$.  
\begin{lemma}  \label{Lemma_A}
The polynomial $A$ has the following properties:\\
(i) $A$ has a common zero with $p_+$ or $p_-$ if and only if the coincidence relation \eqref{cr} holds.\\
(ii) $A < 0$ on $\mathcal{A}^-$ and $A> 0$ on $\mathcal{A}^+$ in the generic case that \eqref{cr} does not
hold.\\
(iii)If \eqref{cr} holds, $A < 0$ on $\mathcal{A}^-$ and $A=0$ on $\mathcal{A}^+=\{\alpha_-^+=\alpha_+^+\}$.
\end{lemma}

\begin{proof}
To prove (i), we express $A$ as a quadratic form in $p^\pm \equiv p_2 \pm v p_1$.
Writing \eqref{A} as
$ \label{A_quad-form_0}
A  = a_1 p_1^2 + 2 a_2 p_1 p_2 + a_3 p_2^2
$
with
$
a_1 = - \big(4 \hat\chi \alpha^2  + 2 (\hat\chi +1) \gamma \alpha + a\big),
a_2 =  (\hat\chi +1) \gamma \alpha + a,
a_3 =   -  a,
$
and substituting now $vp_1 = (p_+ - p_-)/2$ and $p_2 =(p_+ + p_-)/2$ gives
\begin{equation}   \label{A_quad_form}
4v^2 A = A^+ (p^+)^2 + 2 C p^+ p^- + A^- (p^-)^2,
\end{equation}
with
\begin{eqnarray}
A^\pm =  a_1 \pm 2v a_2 +v^2 a_3
\hspace{.8cm} \text{and} \hspace{.8cm}
C = v^2 a_3 - a_1 .
\end{eqnarray}
This directly shows that $A$ vanishes either at a common zero of $p^+$ and $p^-$ --  whose existence is equivalent to \eqref{cr} by Lemma \ref{Lemma_admissibility-region} (i) --, or at common zeros of $A^\pm$ and $p^\mp$. However, neither $A^+$ and $p^-$, nor $A^-$ and $p^+$ have common zeros, since the resultants
$$
\begin{aligned}
\text{res}(A^\pm,p^\mp)&=\left|
\begin{matrix}
-4\hat\chi&2(\pm v-1)(\hat\chi+1)\gamma&-a(1\mp v)^2&0\\
0&-4\hat\chi&2(\pm v-1)(\hat\chi+1)\gamma&-a(1\mp v)^2\\
4\hat\chi&(3\mp v)(\hat\chi+1)\gamma&2(b\mp va)&0\\
0&4\hat\chi&(3\mp v)(\hat\chi+1)\gamma&2(b\mp va)
\end{matrix}
\right|\\
&=-4\hat\chi\left((1+3c_s^2)\hat\chi^2-2(2c_s^4+c_s^2+1)\hat\chi+(1+3c_s^2)\right)
\end{aligned}
$$
do not vanish for $\hat\chi >0$. This proves (i).

We now prove part (ii). For this we first verify (ii) in the special case
that $\hat\chi=1,v=0,c_s^2\neq 1/2$. We then extend this to the general parameter case by a continuity argument. So assume
$\hat\chi=1,v=0,c_s^2\neq \frac12$.  In this case,
$$
A(\alpha)= 4  B(\alpha)-k\quad\text{with }
B(\alpha)=(\alpha^2+2\alpha+1)\alpha^2
\text{ and }
k=c_s^2(1-c_s^2),$$
and
$$
p^\pm(\alpha)   = 4\alpha^2 + 6 \alpha + 2=2(\alpha+1)(2\alpha+1).
$$
At the zeros of $p^\pm$, we find that  
$$
A(-1)=-k< 0
\hspace{1cm} \text{and}  \hspace{1cm}
A(-1/2)=1/4-c_s^2(1-c_s^2)>0.
$$
Since for $v=0$ the intervals $\mathcal{A}^-$ and  $\mathcal{A}^+$ collapse to the two zeros of
$p^+\equiv p^-$, we showed that $A$ is strictly negative on $\mathcal{A}^-=\{-1\}$ and strictly
positive on $\mathcal{A}^+=\{-\frac12\}$. This proves (ii) in the special case $\hat\chi=1,v=0$
for $c_s^2\neq 1/2$.

For $(\hat\chi,v)=(1,0)$ and $c_s^2\neq1/2$, the shape of $A$  is thus as follows:
As $A$ is of degree 4 with $A(\pm\infty)>0$, negative on $\mathcal A^-$, positive on $\mathcal A^+$
and has $A(0)=-k<0$, it has four real zeros, namely
one zero to the left of $\mathcal{A}^-$, a second one between $\mathcal{A}^-$ and
$\mathcal{A}^+$, a third one between $\mathcal{A}^+$ and $0$, and a forth in $(0,\infty)$.
Now, according to (i), this whole situation, i.e., the signs of $A$ on $\mathcal A^\pm$ and
the number and positioning of real roots of $A$
relative to $\mathcal A^-,\mathcal A^+,\{0\}$
is robust against changes of the parameters $(c_s^2,\hat\chi,v)$ in
$X=(0,1)\times(0,\infty)\times(-1,1)$ as long as these avoid the
surface $\Gamma_0=\{\Gamma=0\}\subset X$. As $\Gamma_0$ intersects the line segment
$(0,1)\times\{(1,0)\}$ exactly at the point $(c_s^2,1,0)$ with $c_s^2=1/2$,
the situation, in particular the sign behavior of $A$ on $\mathcal A^\pm$,
thus continues to all parameter values $(c_s^2,\hat\chi,v)\in X\setminus \Gamma_0$,
in either of the two connected components $\Gamma_0$ partitions $X\setminus \Gamma_0$.

(iii) For parameter values with $\Gamma(c_s^2,\hat\chi,v)=0$, the interval $\mathcal{A}^+$ collapses to a
single point. By \eqref{A_quad_form}, $A$ vanishes at this point because it is a common zero of $p^-$
and $p^+$.
\end{proof}

\begin{lemma} \label{Lemma_B}
The polynomial $B$ has the following properties:\\
(i) $B$ has a common zero with $p^+$ or $p^-$ if and only if $\hat\chi =1$ and $v=0$.\\
(ii) $B \leq 0$ on $\mathcal{A}^-$ and $B >0$ on $\mathcal{A}^+$.
\end{lemma}

\begin{proof}
(i) By symmetry of $B$ and $p_\pm$ with respect of $\pm v$, it suffices to
consider only $p_+$,
\begin{equation}  \label{no_common_zeros_eqn1}
p_+(\alpha)  = 4 \hat\chi\alpha^2 + (3+ v)(\hat\chi+1) \gamma \alpha  + 2(b + va).
\end{equation}
As $B(\alpha)=\alpha^2 \tilde B(\alpha)$ with
\be\label{no_common_zeros_eqn2}
\tilde B(\alpha) = \hat\chi \alpha^2  + (\hat\chi +1) \gamma  \alpha   + b
\ee
and $p_+(0)\neq 0$, we are thus looking for common zeros
of $p_+$ and $\tilde B$.
Since
\begin{eqnarray} \label{no_common_zeros_eqn2a}
(p^+ - 4 \tilde B)(\alpha)
&=& (v-1)(\hat\chi+1) \gamma \alpha  - 2(b -va) ,\cr
(b p^+ - 2(b + va) \tilde B)(\alpha)
&=& 2(b -va)  \hat\chi\alpha^2  + \big(b + v(b - 2a) \big)(\hat\chi+1) \gamma \alpha,
\end{eqnarray}
the existence of a common zero would imply that
\begin{eqnarray} \nonumber
\left|\begin{array}{cc}
(v-1)(\hat\chi+1) \gamma & - 2(b -va) \cr  2(b -va)  \hat\chi  &   \big(b + v(b - 2a) \big)(\hat\chi+1) \gamma
\end{array}\right|=0
\end{eqnarray}
or, equivalently,
\begin{equation}  \label{no_common_zeros_eqn3}
 \frac{(\hat\chi+1)^2}{4\hat\chi}=  \frac{(1+c_s^2v)^2}{1+v(v+2)c_s^2}.
\end{equation}
Note now that the left hand side of \eqref{no_common_zeros_eqn3} is  $\ge 1$
with equality holding only for $\hat\chi=1$.
On the other hand, as
\be \nonumber
(1+ c_s^2v)^2 - \big(1+v(v+2)c_s^2\big) = v^2 c_s^2 (c_s^2 -1)\le 0\quad\text{with ``='' only for }v=0,
\ee
 the expression on the right hand side of \eqref{no_common_zeros_eqn3} is
$\le 1$ with equality holding only for $v=0$.

Regarding (ii), obviously $B>0$ on $\mathcal A^+$ since $4\gamma^2 B=A+k$ with $k>0$ by \eqref{A}.
We prove $B\le 0$ on $\mathcal A^-$ first in the special case $v=0$ and then treat the
general case $v\in (-1,1)$ by a continuity argument. Assume that $v=0$ and $\hat\chi\neq 1$
(as the case $v=0$ and $\hat\chi=1$ is already covered by assertion (i)).
By (iii) of Lemma \ref{Lemma_admissibility-region},
$\mathcal{A}^-$ consists of the single point $\alpha^-\equiv\alpha_-^-=\alpha_+^-$, and $\tilde B$, now given
by
\begin{equation} \nonumber
\tilde B (\alpha) = \hat\chi \alpha^2  +   (\hat\chi +1) \alpha   + 1,
\end{equation}
has the zeros
$-{1}/{\hat\chi}$ and $ -1$, at which
\begin{eqnarray} \label{B_sign_eqn1}
p_+\left(-\frac{1}{\hat\chi}\right) = \frac{1}{\hat\chi} -1
\hspace{1cm} \text{and} \hspace{1cm}
p_+(-1) =   \hat\chi -1.
\end{eqnarray}
In particular, $p_+$ has different signs at the two zeros of $\tilde B$. Correspondingly,
$\alpha^-$ lies between these zeros, and $\tilde B(\alpha^-)<0$.
This proves (ii) in the case $v=0$, and by continuity, this also holds true for small
variations of $v$ around $v=0$. By continuity, upon variation of $(v,\hat\chi)$ in the connected domain
$(-1,1)\times(0,\infty)\setminus \{(0,1)\}$, the only way for $B$ to become positive on $\mathcal A^-$
would be to share a zero with $p_-$ or $p_+$; this would contradict property (i).
\end{proof}

Lemmas \ref{Lemma_A} and \ref{Lemma_B} readily imply the assertion of Proposition
\ref{nonlum_L_chi>0} except in the case when \eqref{cr} holds; analogously to the
the situation for $\hat\chi=0$, this case is covered by the following observation,
which is proved in the same way as Lemma \ref{sepcaseforchi=0}.
\begin{lemma} \label{sepcaseforchi>0}
In case \eqref{cr} holds, one finds ${\rm Re}(\tilde\Pi_v^L(\alpha_*+iw,\tilde\bxi))\neq 0$.
\end{lemma}

\section{Mode stability for non-barotropic fluids}
\setcounter{equation}0
In this section, we consider modes in viscous, heat conductive flow
of general
non-barotropic fluids. A non-barotropic fluid is given by an equation of state
$p=p(\theta,\psi)$ that represents the
pressure as a function of temperature and thermodynamic potential; in terms of these the inviscid
parts of the energy-momentum tensor and the particle flux are
\be\label{TN}
T^{\alpha\beta}=\theta^3\frac{\partial p(\theta,\psi)}{\partial \theta}\psi^\alpha\psi^\beta
+p(\theta,\psi)g^{\alpha\beta},
\quad
N^\beta=\frac{\partial p(\theta,\psi)}{\partial \psi}\psi^\beta,
\ee
written in the Godunov variables $\psi^\alpha=u^\alpha/\theta$ and $\psi=(\rho+p)/(n\theta)-s$, where
$n,s$ denote the particle number density and specific entropy.
The Navier-Stokes-Fourier model proposed in \cite{FT2} is
\be\label{nsf}
\frac{\partial}{\partial x^{\beta}}\left(T^{\alpha\beta}+\Delta T^{\alpha\beta}_\Box\right)=0,
\quad
\frac{\partial}{\partial x^{\beta}}\left(N^\beta+\Delta N^\beta_\Box\right)=0
\ee
with $\Delta T^{\alpha\beta}_\Box$ from \eqref{DeltaT} and
\be\label{DeltaN}
-\Delta N^\beta
\equiv
\kappa g^{\beta\gamma}{\partial \psi\over \partial x^\gamma}
+\tilde\sigma\left(U^\beta \Pi^{\gamma\delta}-U^\delta\Pi^{\beta\gamma}\right)
{\partial U_\gamma\over \partial x^\delta},
\ee
where $\tilde\sigma=\sigma n/(\rho+p)$ and $\kappa$ is the coefficient of thermal conductivity.

Subsuming $\psi^\alpha$ and $\psi$ under $\psi^a, a=0,1,2,3,4,$ with $\psi^4=\psi$,
$T^{\alpha\beta}$ and $N^\beta$ as $T^{a\beta}$, and
$\Delta T^{\alpha\beta}_\Box$ and $\Delta N^\beta_\Box$ as $\Delta T^{a\beta}_\Box$,
the equations can be written concisely  as
\be\label{lin}
A^{a\beta g}\frac{\partial \psi_g}{\partial x^\beta}
=
B^{a\beta g\delta}_\Box\frac{\partial^2\psi_g}{\partial x^\beta\partial x^\delta};
\ee
with coefficients
$$
A^{a\beta g}=\frac{\partial T^{a\beta}}{\partial\psi_g},\quad a,g=0,1,2,3,4,
$$
representing the Euler part div$(T)$ and $B^{a\beta g\delta}_\Box$
the Navier-Stokes-Fourier part div$(-\Delta T_\Box)$.
With respect to the rest frame, the coefficients are given (cf.\ \cite{FT2,F20})
by the matrices
$$
A^{a0g}\equiv\begin{pmatrix}
\theta^3p_{\theta\theta}&0&\theta p_{\psi\theta}-p_\psi\\
0&\theta^2p_\theta I_3 &0\\
\theta p_{\psi\theta}-p_\psi&0&\theta^{-1}p_{\psi\psi}
\end{pmatrix},\quad
A^{ajg}\equiv
\begin{pmatrix}
0&\theta^2 p_\theta e_j^\top&0\\
\theta^2 p_\theta e_j&0&p_\psi e_j\\
0&p_\psi e_j^\top&0
\end{pmatrix},
$$
and  
$$
B^{a0g0}_\Box=
\begin{pmatrix}
0&0&0\\
0&-\sigma I_3&0\\
0&0&-\kappa
\end{pmatrix},
\quad
B^{ajgk}_\Box=
\begin{pmatrix}
0&0&0\\
0&
\eta I_3
+(\frac13\eta+\check\zeta) e_j e_k^\top&0\\
0&0&\kappa
\end{pmatrix},
\quad B^{ajg0}_\Box=0_{5}.
$$
The purpose of this section is to prove
\begin{theo}
Assume the fluid has speed of sound $0<c_s<1$. Then for system \eqref{nsf} with
\eqref{DeltaT} and \eqref{DeltaN} and
\be\label{etahatzetakappa}
\eta>0,\quad \check\zeta\ge -\frac13\eta,\quad \sigma=\frac43\eta+\check\zeta,\quad\kappa\geq 0,
\ee
all non-trivial Fourier-Laplace modes with respect to any Lorentz frame are decaying.
\end{theo}

For the case that in addition to viscosity and heat conduction the dissipation also includes diffusion, an analogue of Theorem 3 is implicitly contained in Sec.\ 3 of  \cite{S2}. Note however that formally Theorem 3 does not readily follow from that analogue as the limit of vanishing diffusion is singular.

\begin{proof}
Using direct analogues of Proposition 1 and Lemma 1, we see that it is sufficient to
show the assertion for rest-frame modes. We can thus also assume w.\ l.\ o.\ g.\ that
the wave vector $\bxi$ is of the form $(\xi,0,0)$ with $\xi \ge 0$.

The dispersion relation $0=\Pi(\lambda,(\xi,0,0))$ factors again as $\Pi=\Pi^L(\Pi^T)^2$ with
longitudinal part
$$
\Pi^L(\lambda,\xi)=
\left|\begin{matrix}
\lambda a_1&i\xi a_3&\lambda a_2\\
i\xi a_3&\lambda a_3+\sigma(\lambda^2+\xi^2)&i\xi a_4\\
\lambda a_2&i\xi a_4&\lambda a_5+\kappa(\lambda^2+\xi^2)
\end{matrix}\right|
$$
and transverse part 
$$
\Pi^T(\lambda,\xi)=a_3\lambda+\sigma\lambda^2+\eta\xi^2
$$
with  
$
a_1=\theta^3p_{\theta\theta},\
a_2=\theta p_{\psi\theta}-p_\psi, \
a_3=\theta^2p_\theta, \
a_4=p_\psi, \
a_5=\theta^{-1}p_{\psi\psi}.
$
Regarding the transverse part, the assertion follows in exactly the same way as in the barotropic case in \cite{FT1,FT3}. To study the longitudinal part, we write
\be \nonumber
\Pi^L(\lambda,\xi) = \lambda ^3\: f_1 + \lambda^2 (\lambda^2 + \xi^2)\: f_2 + \lambda \xi^2\: f_3 + \xi^2 (\lambda^2 + \xi^2)\: f_4 + \lambda (\lambda^2 + \xi^2)^2 \: f_5
\ee
with
\begin{eqnarray}
f_1 &=& a_3 (a_1 a_5 - a_2^2) >0 ,    \cr
f_2 &=&  \sigma (a_1 a_5 - a_2^2) + \kappa a_3 a_1  >0  , \cr
f_3 &=&   a_3^2 a_5 - 2 a_3 a_2 a_4 + a_1 a_4^2 >0 ,  \cr
f_4 &=&   \kappa a_3^2 \geq 0 ,   \cr
f_5 &=&  \sigma \kappa a_1 \geq 0  ,
\end{eqnarray}
where 
the signs of the coefficients are
determined as follows: The positive definiteness of $A^{a0g}$ implies positivity of the diagonal
entries $a_1, a_3, a_5 >0$ and of the determinant $a_1 a_5 - a_2^2>0$. This directly implies
$f_1>0$, $f_5\geq 0$ and $f_2 = (f_4 a_1 +\sigma f_1)/a_3>0$; clearly $f_4\geq 0$.
Lastly $f_3$ has the same sign as $f_1$ by the dispersion relation of the inviscid (Euler) part,
\be \label{dispersion_Euler}
\pi^L(\lambda,\xi)= \lambda ^3 f_1  + \lambda \xi^2 f_3 =0,
\ee
which tells us that ${f_3}/{f_1} = c_s^2$. Now, to prove decay, we assume for contradiction that
\be \label{contradiction_non-bar}
\Pi^L(\lambda,\xi)= 0 \hspace{1cm} \text{for} \hspace{.5cm} \lambda = i\beta,
\ee
and compute
\begin{eqnarray}
{\rm Re}(\Pi^L(i\beta,\xi))  &=&  (\xi^2 - \beta^2) \big( \xi^2 f_4 - \beta^2 f_2 \big),   \cr
{\rm Im}(\Pi^L(i\beta,\xi))  &=& \beta \Big( \xi^2 f_3 - \beta^2 f_1 + (\xi^2 - \beta^2)^2 f_5 \Big) .
\end{eqnarray}
Solving the real part of \eqref{contradiction_non-bar} hence implies either $\xi^2 = \beta^2$ or
$\beta^2 = \xi^2f_4/f_2$. In the first case, the imaginary part of
\eqref{contradiction_non-bar} would give $f_3=f_1$,
which contradicts $c_s^2<1$.  In the second case, the imaginary part of \eqref{contradiction_non-bar} would imply
\be
f_2 f_3 - f_4 f_1 + \xi^2 \big(1-  f_4/f_2 \big)^2 f_5 =0,
\ee
a contradiction since $f_2 f_3 - f_4 f_1 = ({1}/{a_3})\big(\sigma f_1 f_3 + f_4 (a_1 a_4 - a_2 a_3)^2\big) >0$.
We have thus proved that Re$(\lambda)\neq 0$ for all modes.        

A trivial computation in the spirit of the proof of Lemma  6.6 in \cite{FT1} shows that Re$(\lambda)<0$
for small $\xi>0$. As Re$(\lambda)$ cannot reach $0$, the assertion follows by continuity.
\end{proof}

\end{document}